\documentclass[aoas]{imsart}

\usepackage{amsthm,amsmath, amsfonts}
\RequirePackage[OT1]{fontenc}
\RequirePackage{amsthm,amsmath}
\RequirePackage{natbib}
\RequirePackage[colorlinks,citecolor=blue,urlcolor=blue]{hyperref}
\usepackage[latin1] {inputenc}
\usepackage{vmargin}
\usepackage{multirow}

\usepackage{tikz}
\usetikzlibrary{shapes,snakes}
\usetikzlibrary{calc, arrows, decorations.markings}

\usepackage{comment}
\usepackage{amsthm, dsfont, tipa}
\usepackage{graphics}
\usepackage{bbm}
\usepackage{bm}
\newlength{\tempdim}
\usepackage{graphicx,float,multirow,subfigure}
\usepackage{tikz} 
 \tikzset{
  big arrow/.style={
    decoration={markings,mark=at position 1 with {\arrow[scale=2,#1]{>}}},
    postaction={decorate},
    shorten >=0.4pt},
  big arrow/.default=black}
  
  \tikzset{
  dash arrow/.style={
    decoration={markings,mark=at position 1 with {\arrow[scale=1.2,#1]{>}}},
    postaction={decorate},
    shorten >=0.4pt},
  dash arrow/.default=black}

\theoremstyle{plain}

\newtheorem{lemma}{Lemma}[section]
\newtheorem{corollary}{Corollary}[section]

\usepackage{amssymb}
\usepackage{epsfig}
\startlocaldefs

\newtheorem{mydef}{Definition}

\newcommand{\LR}{\textrm{LR}}

\endlocaldefs

\arxiv{stat.AP/1502.02406}

\begin{document}

\setattribute{journal}{name}{}

\begin{frontmatter}

\title{Nonparametric Bayesian approach to LR assessment in case of rare \textrm{type} match}
\runtitle{Nonparametric Bayesian approach to LR assessment}

\begin{aug}
  \author{\fnms{Giulia}  \snm{Cereda}\corref{}\ead[label=e1]{giulia.cereda@unil.ch}}
 % \author{\fnms{Second} \snm{Author}\ead[label=e2]{second@somewhere.com}}
  %\and
  %\author{\fnms{Third}  \snm{Author}%
 % \ead[label=e3]{third@somewhere.com}%
  %\ead[label=u1,url]{http://www.foo.com}}

 % \thankstext{t2}{Footnote to the first author with the `thankstext' command.}

  \runauthor{Giulia Cereda}

  \affiliation{University of Lausanne}
  \and
  \affiliation{Leiden University}
  \address{University of Lausanne, Ecole des sciences criminelles, \\ Facult\'e de droit, des sciences criminelles et d'administration publique\\
1015, Lausanne-Dorigny, Switzerland,\\ 
          \printead{e1}}

\end{aug}
\begin{abstract}

The evaluation of a match between the DNA profile of a stain found on a crime scene and that of a suspect (previously identified) involves the use of the unknown parameter $\mathbf{p}=(p_1, p_2, ...)$, (the ordered vector which represents the frequencies of the different DNA profiles in the population of potential donors) and the names of the different DNA types.
We propose a Bayesian nonparametric method which models $\mathbf{p}$ through a random variable $\mathbf{P}$ distributed according to the two-parameter Poisson Dirichlet distribution, and discards the information about the names of the different DNA types.
The ultimate goal of this model is to evaluate the so-called `probative value' of DNA matches in the rare type case, that is the situation in which the suspect's profile, matching the crime stain profile, is not in the database of reference.
\end{abstract}
\end{frontmatter}

\section{Introduction}

The largely accepted method for evaluating how much some available data $\mathcal{D}$ (typically forensic evidence) is helpful in discriminating between two hypotheses of interest (the prosecution hypothesis $H_p$ and the defense hypothesis $H_d$), is the calculation of the \emph{likelihood ratio} (LR), a statistic that expresses the relative plausibility of the data under these hypotheses, defined as
\begin{equation}
\label{eqa}
\LR=\frac{\Pr(\mathcal{D}|H_p)}{\Pr(\mathcal{D}|H_d)}.
\end{equation}

Widely considered the most appropriate framework to report a measure of the `probative value' of the evidence regarding the two hypotheses \citep{robertson:1995, evett:1998, aitken:2004,balding:2005}, it indicates the extent to which data is in favor of one hypothesis over the other. 
Forensic literature presents many approaches to calculate the LR, mostly divided into Bayesian and frequentist methods (see \citet{cereda:2015b} for a careful differentiation between these two approaches). 

This paper proposes a Bayesian nonparametric method for the LR assessment in the rare type match case, 
the challenging situation in which there is a match between some characteristic of the recovered material and of the control material, but this characteristic has not been observed yet in previously collected samples (i.e. database of reference). 
This constitutes a problem because the LR value depends on the proportion of the matching characteristic in a reference population, and this proportion is, in standard practice, estimated using the relative frequency of the characteristic in the available database.
In particular, we will focus on Y-STR DNA profile matches, for which the rare type match problem is often recurring \citep{cereda:2015b}. In this case data to evaluate is made of the information about the matching profile and of the list of DNA profile in the database.

The use of a Bayesian nonparametric method is justified by the fact that the parameter of the model is the infinite dimensional vector $\mathbf{p}$, made of the (unknown) sorted population proportions of all possible Y-STR profiles, assumed to be infinitely many. As prior over $\mathbf{p}$ we choose the two parameter Poisson Dirichlet distribution, and treat its parameters as hyperparameters, hence provided with a hyperprior.
Moreover, we will discard the information contained in the names of the profiles, and this will lead to a reduction of the data $\mathcal{D}$ to  a smaller amount of information $D$.
The reduction of the data can be a wise practice in presence of many nuisance parameters as explained in \citet{cereda:2015b}, and  sometimes the likelihood ratio based on the data reduction is much more precisely estimated than the likelihood ratio based on all data.

The paper is structured in the following way: Section~\ref{bnpm} introduces the notation, the assumptions of our model and the prior distribution chosen for parameter $\mathbf{p}$. Section~\ref{m} displays the model, via Bayesian network representation, along with some theory on random partitions useful to define a clever and compact representation of the reduced data $D$. 
An alternative representation of the same model via the Chinese restaurant process is also described.
Section~\ref{gm} introduces relevant known results regarding the two parameter Poisson Dirichlet distribution, along with a new lemma, which can be used for all the situations in which prosecution and defense agree on the distribution of part of the data and disagree on the distribution of the rest, given the parameter(s). This result will allow to derive the LR in a very elegant way (Section~\ref{lr}).
Section~\ref{ard} displays some experiments of application of this model on a real database of Y-STR profiles, such as model fitting, asymptotic power law behavior, study of the loglikelihood function, and comparison with the LR values obtained in the ideal situation in which vector $\mathbf{p}$ is known. Lastly, Section~\ref{frq} proposes questions for future research.
\section{A Bayesian nonparametric model for the rare type match}\label{bnpm}

\subsection{The rare type match problem} \label{rare}

In order to evaluate the match between the profile of a particular piece of evidence and a suspect's profile, it is necessary to estimate the proportion of that profile in the population of potential perpetrators. Indeed, it is intuitive that the rarer the matching profile, the more the suspect is in trouble. 
Problems arise when the observed frequency of the profile in a sample from the population of interest (i.e., in a reference database) is 0. Such characteristic is likely to be rare, but it is challenging to quantify how rare it is. The rare type problem is particularly important in case a new kind of forensic evidence, such as results from DIP-STR markers (see for instance \citet{cereda:2014b}) is involved, and for which the available database size is still limited.
The same happens when Y-chromosome
(or mitochondrial) DNA profiles are used since the set of possible Y-STR profiles is extremely
large. As a consequence, most of the Y-STR haplotypes are not represented in the database.
The Y-STR marker system will thus be retained here as an extreme but in practice common and important way in
which the problem of assessing evidential value of rare type match can arise. This problem is so substantial that it has been defined ``the fundamental problem of forensic mathematics'' \citep{brenner:2010}.

The \emph{empirical frequency estimator}, also called \emph{naive estimator}, that uses the frequency of the characteristic in the database, puts unit probability mass on the set of already observed characteristics, and it is thus unprepared for the observation of a new type. 
A solution could be the \emph{add-constant} estimators (in particular the well known \emph{add-one} estimator, due to \citet{laplace:1814}, and the \emph{add-half} estimator of \citet{krichevsky:1981}), which add a constant to the count of each type, included the unseen ones. However, this method requires to know the number of possible unseen types, and it is also not very performing when this number is large compared to the sample size (see \citet{gale:1994} for additional discussion). 
Alternatively, \citet{good:1953}, based on an intuition on A.M. Turing, proposed the \emph{Good Turing estimator} for the total unobserved probability mass, based on the proportion of singleton observations in the sample. An extension of this estimator is applied to the LR assessment in the rare type match in \citet{cereda:2015b}. 
For a comparison between \emph{add one} and \emph{Good-Turing} estimator, see \citet{orlitsky:2003}.
As pointed out in \citet{anevski:2013}, the \emph{naive estimator}, and the \emph{Good Turing estimator} are
in some sense complementary: the first gives a good estimate for the observed types, and the second for the probability mass of the unobserved ones. 
More recently, \citet{orlitsky:2004} have introduced the \emph{high profile estimator}, which extends the tail of the \emph{naive estimator} to the region of unobserved types. \citet{anevski:2013} improved this estimator and provided the consistency proof. 
Papers that address the rare Y-STR haplotype problem in forensic context are for instance \citet{egeland:2008}, \citet{brenner:2010}, and \citet{cereda:2015}, which applied classical Bayesian approach (the beta binomial and the Dirichlet multinomial problem) to the LR assessment in the rare haplotype case. Moreover, the Discrete Laplace method presented in \citet{andersen:2013b}, even though not specifically designed for the rare type case can be successfully applied to that extent \citet{cereda:2015b}.

Bayesian nonparametric estimators for the probability of observing a new type have been proposed by \citet{tiwari:1989}, using Dirichlet process, by \citet{lijoi:2007} using general Gibbs prior, and by \citet{favaro:2009} with specific interest to the two parameter Poisson Dirichlet prior. 
However, for the LR assessment it is required not only the probability of observing a new species but also the probability of observing this same species twice (according to the defense the crime stain profile and the suspect profile are two independent observations), and to our knowledge, this paper is the first one to address the problem of LR assessment in the rare haplotype case using Bayesian nonparametric models. As prior we will use the Poisson Dirichlet distribution, which is proving useful in many discrete domain, in particular language modelling \citep{teh:2006}. In addition, it shows a power law behaviour which describe a incredible variety of phenomena \citep{newman:2005}.

 \subsection{Notation}
Throughout the paper the following notation is chosen: random variables and their values are denoted, respectively, with uppercase and lowercase characters: $x$ is a realization of $X$. Random vectors and their values are denoted, respectively, by uppercase and lowercase bold characters: $\mathbf{p}$ is a realization of the random vector $\mathbf{P}$. Probability is denoted with $\Pr(\cdot)$, while density of a continuous random variable $X$ is denoted alternatively by $p_{X}(x)$ or by $p(x)$ when the subset is clear from the context. For a discrete random variable $Y$, the density notation $p_Y(y)$ and the discrete one $\Pr(Y=y)$ will be alternately used.
Moreover, we will use shorthand notation like $p(y \mid x)$ to stand for the probability density of Y with respect to the conditional distribution of $Y$ given $X = x$.

Notice that when using the notation in~\eqref{eqa}, $\mathcal{D}$ is regarded as events. However, later in the paper, it will be regarded as a random variables. In that case, the following notation will thus be preferred:
\begin{equation}
\label{eq21rr}
\LR=\frac{\Pr(\mathcal{D}=d|H_p)}{\Pr(\mathcal{D}=d|H_d)}\quad \text{or} \quad \frac{p(d|H_p)}{p(d|H_d)}.
\end{equation}

Lastly, notice that ``DNA types'' is used throughout the paper as a general formula to indicate Y-STR profiles.
\subsection{Model assumptions}\label{m-a}

Our model is based on the two following assumptions:

\begin{description}
\item[Assumption 1] There are infinitely many DNA types in Nature.
\end{description}
The reason for this assumption is that there are so many possible DNA types that they can be considered infinite. This assumption, already used by e.g.\ \citet{kimura:1964} in the `infinite alleles model', allows to use Bayesian nonparametric methods and avoids the problem of specifying how many different types there are in Nature.
\begin{description}

\item[Assumption 2] The names of the different DNA types do not contain information.

\end{description}
The specific sequence of numbers that forms a DNA profile carries information: if two profiles show few differences this means that they are separated by few mutation drifts, hence the profiles share a relatively recent common ancestor. However, this information is difficult to exploit and may be not so relevant for the LR assessment. 
This is the reason why we will treat DNA types as ``colors'', and only consider the repartition into different categories. Stated otherwise, we put no topological structure on the space of the DNA types. 

Notice that this assumption makes the model a priori suitable for any characteristic which shows many different possible types, thus what written still holds, in principle, also replacing `DNA types' with any other type. However, we will only test the model with Y-STR data.

\subsection{Prior}

In Bayesian statistics, parameter(s) of interest are modeled through random variables.
The (prior) distribution over the parameter(s) should represent the uncertainty about its (their) value(s).

LR assessment for the rare type match involves two unknown parameters of interest: one is $h \in \{H_p, H_d\}$, representing the unknown true hypothesis, the other is $\boldsymbol{p}$, the vector of the unknown population frequencies of all DNA profiles in the population of potential perpetrators.
The dichotomous random variable $H$ is used to model parameter $h$, and the posterior distribution of this random variable, given data, is the ultimate aim of a forensic inquiry.
In a similar way, random variable $\boldsymbol{P}$ can be used to model $\boldsymbol{p}$. 
Because of Assumption~1, $\boldsymbol{p}$ is an infinite dimensional parameter, hence the need of Bayesian nonparametric methods \citep{hjort:2010}.
In particular $\boldsymbol{p}= (p_t | t \in T)$, with $T$ a countable set of indexes, $p_{t}>0$, and $\sum_{t}p_t=1$.
Moreover, because of Assumption~2, data will be reduced to random partitions, as explained in Section~\ref{partitions}, and it will turn out that the distribution of these partitions does not depend on the order of the $p_i$. Hence, we can force the parameter $\mathbf{p}$ to have values in $\nabla_{\infty}=\{(p_1, p_2, ...) | p_1\geq p_2 \geq ..., \sum p_i=1, p_i>0\}$, the ordered infinite dimensional simplex.
The ordered random vector $\mathbf{p}$ describes an infinite population randomly partitioned into DNA types. The randomness is described by the prior distribution over $\mathbf{p}$, for which we choose the two-parameter Poisson Dirichlet distribution \citep{ pitman:1997,  feng:2010,  buntine:2012, carlton:1999, pitman:2006}, defined in the following way: 
\begin{mydef}[two parameter GEM distribution]
Given $\alpha$ and $\theta$ satisfying the following conditions:
\begin{equation}\label{condition}
0\leq \alpha<1, \text{ and } \theta>-\alpha.
\end{equation}
the vector $\bm{W}=(W_1,W_2,...)$ is said to be distributed according to the \emph{GEM($\alpha, \theta$)}, if $$\forall i\quad W_i=V_i\prod_{j=1}^{i-1} (1-V_j),$$ where $V_1$, $V_2$,... are independent random variables distributed according to $$V_i\sim \text{B}(1-\alpha, \theta+ i\alpha).$$ 
It holds that $W_i >0$, and $\sum_{i}W_i=1.$ 
\end{mydef}
The GEM distribution (short for Griffin - Engen - McCloskey distribution') is well known in literature as the ``stick breaking prior'', since it measures the random sizes in which a stick is broken iteratively. 
This distribution is invariant under size biased permutation \citep{engen:1975}, the random permutation defined by sampling from the population and assigning to each type a label, based on the order in which it is first sampled.

\begin{mydef}[Two parameter Poisson Dirichlet distribution]\label{pd}
Given $\alpha$ and $\theta$ satisfying condition~\eqref{condition}, and a vector $\bm{W}=(W_1,W_2,...) \sim \text{GEM}(\alpha, \theta)$, the random vector $\boldsymbol{P}=(P_1, P_2, ...)$ obtained by ordering $\bm{W}$, such that $p_i\geq p_{i+1}$, is said to be \emph{Poisson Dirichlet distributed} PD$(\alpha, \theta)$. Parameter $\alpha$ is called \emph{discount parameter}, while $\theta$ is the \emph{concentration parameter.}

\end{mydef}

Notice that the vector $\bm{P}$ is obtained by sorting the vector $\bm{W}$ in non increasing order, while the vector $\bm{W}$ can be obtained by the so-called \emph{size biased permutation} of the indexes of $\bm{P}$ \citep{perman:1992, pitman:1997}.

The two parameter Poisson Dirichlet distribution PD$(\alpha, \theta)$ is the generalization of the well known Poisson Dirichlet distribution with a single parameter $\theta$ introduced by \citet{kingman:1975}, which is the representation measure \citep{kingman:1977, kingman:1978} of the celebrated \emph{Ewens sampling formula} \citep{ewens:1972}, widely applied in genetics \citep{karlin:1972, kingman:1980}.
For our model we will not allow $\alpha=0$, hence we will assume $0<\alpha<1$. 

It is worth mentioning that an alternative choice for the parameters space is $\alpha<0$, $\theta=-m\alpha$ for some $m\in\mathbb{N}$ \citep{pitman:1996, gnedin:2006, gnedin:2010, cerquetti:2010}. It corresponds to a model with finitely many ($m$) DNA types, where the prior over $\mathbf{P}=(P_1, ..., P_m)$ is Dirichlet with $m$ parameters equal to $-\alpha$.

Lastly, we point out that, in practice, we cannot assume to know parameters $\alpha$ and $\theta$: this is why we will treat them as hyperparameters on which we will put an hyperprior.

 \section{The model}\label{m}

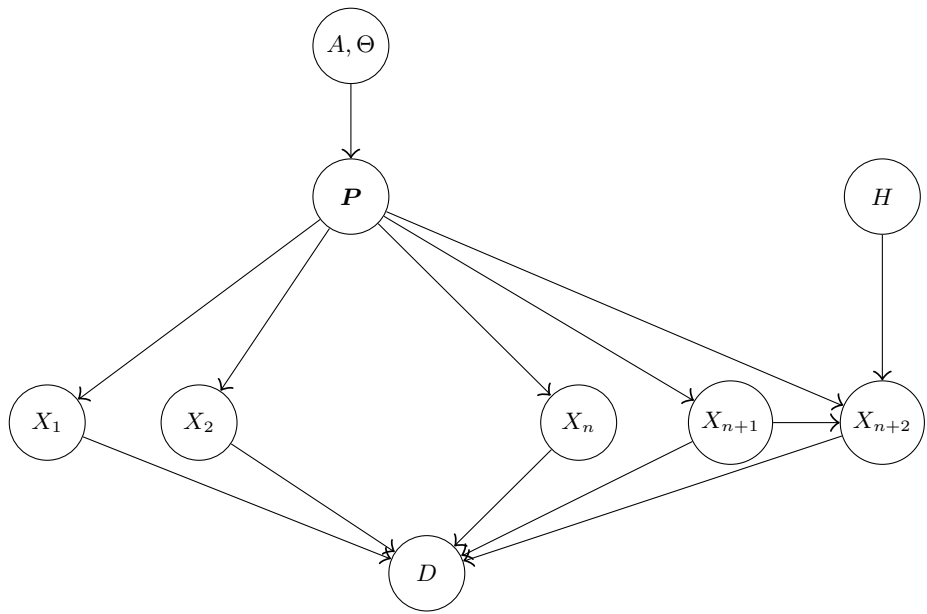
\begin{figure}[htbp]
\begin{center}
\begin{tikzpicture} 
\node[draw, circle, minimum size=1cm]                (t) at (-4,4)  { $A, \Theta$};
\node [draw, circle, minimum size=1cm]              (h) at (3,2) { $H$};
\node [draw, circle, minimum size=1cm]              (p) at (-4,2) { $\bm{P}$};
\node [draw, circle, minimum size=1cm]              (y1) at (-8, -1) { $X_1$};
\node [draw, circle, minimum size=1cm]              (y2) at (-6,-1) { $X_2$};
\node [draw, circle, minimum size=1cm]              (yn) at (-1,-1) { $X_n$};
\node [draw, circle, minimum size=1cm]              (yn33) at (1,-1) { $X_{n+1}$};
\node [draw, circle, minimum size=1cm]              (yn34) at (3,-1) { $X_{n+2}$};
\node [draw, circle, minimum size=1cm]              (d) at (-3,-3) { $D$};

   \draw[black, big arrow] (t) -- (p);
   \draw[black, big arrow] (p) -- (y1); 
   \draw[black, big arrow] (p) -- (y2); 
   \draw[black, big arrow] (p) -- (yn); 
      \draw[black, big arrow] (p) -- (yn33); 
    \draw[black, big arrow] (p) -- (yn34); 
    \draw[black, big arrow] (h) -- (yn34); 
    \draw[black, big arrow] (yn33) -- (yn34); 
  \draw[black, big arrow] (y1) -- (d); 
    \draw[black, big arrow] (y2) -- (d); 
      \draw[black, big arrow] (yn) -- (d); 
        \draw[black, big arrow] (yn33) -- (d); 
  \draw[black, big arrow] (yn34) -- (d);

\end{tikzpicture}
\caption{Bayesian network to show the conditional dependencies of the relevant random variables in our model.}
\label{berhnt}
\end{center}
\end{figure}

The Bayesian network of Figure~\ref{berhnt} encapsulates the conditional dependencies of the variables of the proposed model.
They are defined through random variables defined as follows:

\begin{itemize}

\item $H$ is a dichotomous random variable that represents the hypotheses of interest and can take values $h \in \{H_p,H_d\}$, according to the prosecution or the defense, respectively. A uniform prior on the hypotheses is chosen:
$$\Pr(H=h)\propto 1 \quad \textrm{for} \ h=\{H_p, H_d\}.$$

\item ($A, \Theta$) is the random vector that represents the hyperparameters $\alpha$ and $ \theta$, satisfying condition~\eqref{condition}. The joint distribution of these two parameters (hyperprior) will be generically denoted as $p(\alpha, \theta)$:
$$(A, \Theta) \sim p(\alpha, \theta).$$

\item The random vector $\bm{P}$ with values in $\nabla_{\infty}$, represents the ranked population frequencies. $\bm{P}=\bm{p}=(p_1, p_2, ...)$ means that $p_1$ is the frequency of the most common DNA type in the population, $p_2$ is the frequency of the second most common DNA type, and so on. As a prior for $\bm{P}$ we use the two-parameter Poisson Dirichlet distribution as in Definition~\ref{pd}:
$$\bm{P}| (A, \Theta) = (\alpha, \theta) \sim PD(\alpha, \theta).$$

\item Integer valued random variables $X_1$, ..., $X_n$ represent
the ranks of the population proportions of the DNA types of the individuals in the database (after some arbitrary ordering for profiles in the database is chosen). For instance, $X_3 = 5$ means that the third individual in the database has the fifth most common DNA type in the population.
Since $\bm{p}$ is unknown these random variables cannot be observed. 
Given $\bm{p}$ they are an i.i.d. sample from $\bm{p}$:
\begin{equation} \label{eqx}
X_1, X_2, ..., X_n | \bm{P}=\bm{p} \sim_{i.i.d.} \bm{p}.
\end{equation}

\item $X_{n+1}$ represents the rank of the suspect's DNA type. It is again a draw from $\bm{p}$.

$$ X_{n+1} | \bm{P}=\bm{p} \sim  \bm{p}.$$

\item $X_{n+2}$ represents the rank of the crime stain's DNA type. According to the prosecution, given $X_{n+1}$, this random variable is deterministic (it is equal to $x_{n+1}$ with probability 1). According to the defense it is another sample from $\bm{p}$:

$$X_{n+2} |  \bm{P}=\bm{p} , X_{n+1}=x_{n+1}, H=h \sim 
\begin{cases} 
\delta_{x_{n+1}} & \text{if } h=H_p\\
\bm{p} &  \text{if } h=H_d
\end{cases}.
$$
\end{itemize}

As already mentioned, $X_1, ..., X_{n+2}$ can not be observed. They represent the database ranked according to the unknown rank in $\boldsymbol{p}$ and constitute an intermediate layer that helps in expressing the data in terms of observable partitions. Section~\ref{partitions} recalls some notions about random partitions, useful before defining node $D$, representing the `reduced' data we want to evaluate.

\subsection{Random partitions}\label{partitions}

A \emph{partition of a set $A$} is an unordered collection of nonempty and disjoint subsets of $A$ the union of which forms $A$. 
Particularly interesting for our model are partitions of the set $A=[n]=\{1, ..., n\}$, denoted as $\pi_{[n]}$. The set of all partitions of $[n]$ will be denoted as $\mathcal{P}_{[n]}$. Random partitions of $[n]$ will be denoted as $\Pi_{[n]}$. In addition, a \emph{partition of $n$} is a finite non increasing sequence of positive integers that sum up to $n$. Partitions of $n$ will be denoted as $\pi_n$. 

%
%A \emph{partition of $n$} is a finite non increasing sequence of positive integers that sum up to $n$. Partitions of $n$ will be denoted as $\pi_n$. The set of all partitions of $n$ will be denoted as $\mathcal{P}_{n}$. Random partitions of $n$ will be denoted as $\Pi_{n}$.
%

Given a sequence of integer valued random variables $X_1, ..., X_n$, let $\Pi_{[n]}(X_1, X_2, ..., X_n)$ be the random partition defined by the equivalence classes of their indexes using the random equivalence relation $i \sim j$ iff $X_i=X_j$.
This construction allows to build a map from the set of values of $X_1, ..., X_n$ to the set of the partitions of $[n]$ as in the following example ($n=10$):
\begin{align*}
\mathbb{N}^{10} &\rightarrow  \mathcal{P}_{[10]}  \\
X_1, ..., X_{10}  &\longmapsto  \Pi_{[10]}(X_1, X_2, ..., X_{10}) \\
(3, 1, 3, 1, 2, 2, 6, 9, 4, 1)&\longmapsto \{ \{1,3\}, \{2, 4, 10 \}, \{5, 6\}, \{7\}, \{8\}, \{9\} \} \end{align*}

Typical data to evaluate in case of a match is $\mathcal{D}=(E, B)$, where $E=(E_s, E_t)$, and
\begin{itemize}
\item  $B$ = the database of size $n$, which contains a sample of DNA types, indexed by $i= 1, ..., n$, from the population of possible perpetrators,
\item $E_s$ = suspect's DNA type,
\item $E_t$ = crime stain's DNA type (matching with the suspect's type).
\end{itemize}

In agreement with Assumption~2, we can consider the reduction of data which ignores information about the names of the DNA types: this is achieved, for instance, by retaining only the equivalence classes of the relation ``to have the same DNA type''.   

The database can thus be reduced to the partition of $[n]$, denoted $\pi_{[n]}^{\text{Db}}$, and obtained using the equivalence classes of the indexes. 
Notice that the same partition is obtained via random variables $X_1, ..., X_n$, as defined in \eqref{eqx}.
%and further to the corresponding partition of $n$, denoted as $\pi_{n}^{\text{Db}}$. 
Stated otherwise, we can reduce $B$ to $\pi_{[n]}^{\text{Db}}$, the partition of $[n]$ obtained from the database.
However, data is actually made of the background data along with the evidence, two new observations that match. 
In a similar way, when the suspect profile is considered we obtain the partition $\pi^{\text{Db}+}_{[n+1]}$, where the first $n$ integers are partitioned as in $\pi^{\text{Db}}_{[n]}$, and $n+1$ constitutes a single subset (at least in the rare type match case). 

When the crime stain profile is considered we obtain the partition $\pi^{\text{Db}++}_{[n+2]}$ where the first $n$ integers are partitioned as in  $\pi^{\text{Db}}_{[n]}$, and $n+1$ and $n+2$ belongs to the same (new) subset.

Random variables $\Pi^{\text{Db}}_{[n]}$, $\Pi^{\text{Db}+}_{[n+1]}$, and $\Pi^{\text{Db}++}_{[n+2]}$ are used to model $\pi_{[n]}^{\text{Db}}$, $\pi_{[n+1]}^{\text{Db}+}$, and $\pi_{[n+2]}^{\text{Db}++}$, respectively.

Notice that, given $\alpha$ and $\theta$, prosecution and defense agree on the distribution of $\Pi^{\text{Db}+}_{[n+1]}$ but disagree on the distribution of $\Pi^{\text{Db}++}_{[n+2]}$.

It is worth noticing that, by construction, the same random partitions can be defined through random variables $X_1$, ..., $X_{n+2}$:

\begin{align*}
\Pi_{[n]}^{\text{Db}}&= \Pi_{[n]}(X_1, ..., X_n),\\
\Pi_{[n+1]}^{\text{Db}+}&= \Pi_{[n+1]}(X_1, ..., X_{n+1}),\\
\Pi_{[n+2]}^{\text{Db}++}&= \Pi_{[n+2]}(X_1, ..., X_{n+2}).\\
\end{align*}

To clarify, consider the following example of a database (Db) with $k=6$ different DNA types, from $n=10$ individuals: 
$$\text{Db}=(h_1, h_2, h_1, h_2, h_3, h_3, h_4, h_5, h_6,h_2),$$ where $h_i$ is the name of the $i$th DNA type in the order chosen for the database. 
This can be reduced to the partition of $[10]$: $$\pi_{[10]}^{\text{Db}}=\{\{1,3\}, \{2, 4, 10\}, \{5,6\}, \{7\}, \{8\}, \{9\}\}.$$ 
Then, the part of data prosecution and defense agree on is $$\pi_{[11]}^{\text{Db}+}=\{\{1,3\}, \{2, 4, 10\}, \{5,6\}, \{7\}, \{8\}, \{9\},  \{11\} \},$$ while the entire data $D$ can be represented as  $$\pi_{[12]}^{\text{Db}++}=\{\{1,3\}, \{2, 4, 10\}, \{5,6\}, \{7\}, \{8\}, \{9\},  \{11, 12\} \}.$$

Now, assume that $\boldsymbol{p}$ is known, thus we know also that $h_1$ is, for instance, the fourth most frequent type, $h_2$ is the second most frequent type, and so on.
Stated otherwise, we are now able to observe the variables $X_1$, ..., $X_{n+2}$: $X_1= 4$,  $X_2=2$, $X_3=4$, $X_4=2$, $X_5= 3$, $X_6=3$, $X_7=10$, $X_8=13$, $X_9=5$, $X_{10}=2$, $X_{11}=9$, $X_{12}=9$.
It is easy to check  that $\Pi_{[n]}(X_1, ..., X_n)=\pi_{[n]}^{\text{Db}}$, $\Pi_{[n+1]}(X_1, ..., X_{n+1})=\pi_{[n+1]}^{\text{Db+}}$, $\Pi_{[n+2]}(X_1, ..., X_{n+2})=\pi_{[n+2]}^{\text{Db++}}.$

Data $D$ can now be defined as:

\begin{itemize}
\item $D=\pi^{\text{Db}++}_{[n+2]}$, the  partition of $[n+2]$ obtained from the database enlarged with the two new observations. \end{itemize}

Node $D$ of Figure~\ref{berhnt} is defined accordingly. Notice that, given $X_1,..., X_{n+2}$, $D$ is deterministic. A very relevant result is that, according to Proposition 4 in \citet{pitman:1992b} it is possible to describe directly the distribution of $D\mid \alpha, \theta, H$. Hence, we can get rid of the intermediate layer of nodes $X_1$, ..., $X_{n+2}$. In particular, it holds that if $$\bm{P}\mid \alpha, \theta \sim PD(\alpha, \theta),$$ and $$X_1, X_2, ... \mid \bm{P}=\bm{p} \sim_{\text{i.i.d}} \bm{p},$$ then, for all $n$, the random partition $\Pi_{[n]}=\Pi_{[n]}(X_1, ..., X_{n})$ has the following distribution:
\begin{equation}\label{ChEPPF}
\mathbb{P}_{n}^{\alpha,\theta}(\pi_{[n]}):=\Pr(\Pi_{[n]}=\pi_{[n]}| \alpha, \theta)=\frac{[\theta+ \alpha]_{k-1; \alpha}}{[\theta+ 1]_{n-1; 1}}\prod_{i=1}^k[1-\alpha]_{n_i-1;1}, 
\end{equation}
where $n_i$ is the size of the $i$th block of $\pi_{[n]}$ (the blocks are here ordered according to the least element), and $\forall x, b\in \mathbb{R}, a\in \mathbb{N}$, $[x]_{a,b}:=\begin{cases}
\prod_{i=1}^{a-1} (x+ib) &\text{if } a\in \mathbb{N}\backslash \{0\}\\
0 &\text{if } a=0
\end{cases}.$ This formula is also known as the \emph{Pitman sampling formula}, further studied in \citet{pitman:1995}.
The model of Figure~\ref{berhnt} can thus be simplified (see Figure~\ref{bernt2c}).

%

%
%Notice that to each partition $\pi_{n}= (n_1, ..., n_k)$ correspond $N^{\pi_n}$ different partitions of $[n]$, with $N^{\pi_n}:=\binom{n}{n_1, ..., n_k}\frac{1}{\prod_{i=1}^n a_i!}$, and $a_i=\#\{j| n_j=i\}$ \citep{pitman:1995}. Because of the symmetry, it holds that $\Pr_{\alpha,\theta}(\pi_n= (n_1, ..., n_k))= N^{\pi_n}p_{\alpha, \theta}(n_1, n_2, ..., n_k)$.
%

%For the moment, instead of specifying the distribution of $D^*$ we can just  re-write it in the following convenient form:
%\begin{align}
%&\Pr(D^*=\pi_{[n+2]}| \alpha, \theta)=\\
% &\Pr(\Pi_{[n+2]}=\pi_{[n+2]}| \Pi_{[n+1]}=\pi_{[n+1]}, \alpha, \theta)\Pr(\Pi_{[n+1]}=\pi_{[n+1]}| \alpha, \theta)
% \end{align}

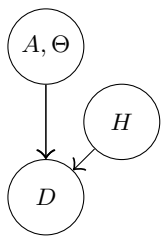
\begin{figure}[htbp]
\begin{center}
\begin{tikzpicture}  \node[draw, circle, minimum size=1cm]                (t) at (-2,0)  { $A, \Theta$};
\node [draw, circle, minimum size=1cm]              (h) at (-1,-1) { $H$};
 \node [draw, circle, minimum size=1cm]              (d) at (-2,-2) { $D$};

 \draw[black, big arrow] (t) -- (d);
 \draw[black, big arrow] (h) -- (d);     

\draw[black, big arrow] (t) -- (d);

\end{tikzpicture}
\caption{Simplified version of the Bayesian network in Figure~\ref{berhnt}}
\label{bernt2c}
\end{center}
\end{figure}
 It holds that $\Pr(D|\alpha, \theta, H_p) = \mathbb{P}_{n+1}^{\alpha, \theta}(\pi_{[n+1]}^{\text{Db+}})$, and $\Pr(D|\alpha, \theta, H_d) = \mathbb{P}_{n+2}^{\alpha, \theta}(\pi_{[n+2]}^{\text{Db++}})$.
 
Notice that for $\alpha=0$ we obtain the Ewens's sampling formula.
\subsection{Chinese Restaurant representation}
There is an alternative characterization of this model, called ``Chinese restaurant process'', due to \citet{aldous:1985} for the one parameter case, and studied in details for the two parameter version in \citet{pitman:2006}. It is defined as follows:
consider a restaurant with infinite tables, each one infinitely large. 
Let $Y_1, Y_2, ...$ be integer valued random variables that represent the seating plan: tables are ranked in order of occupancy, and $Y_i = j$ means that the $i$th customer seats at the $j$th table to be created. 
The process is described by the following transition matrix:
$$Y_1=1$$
\begin{equation} \label{eqch}
\Pr(Y_{n+1} = i | Y_1, ..., Y_n)=\begin{cases} 
{\displaystyle \frac{\theta + k\alpha}{n+\theta} }& \text{if } i = k+1 \\
&\\
{\displaystyle\frac{n_i- \alpha}{n+\theta} }& \text{if } 1\leq i \leq k \\
\end{cases} 
\end{equation}

where $k$ is the number of tables occupied by the first $n$ customers, and $n_i$ is the number of customers that occupy table $i$.

$Y_1, ..., Y_n$ are not i.i.d., nor exchangeable, but it holds that $\Pi_{[n]}(Y_1, ..., Y_n)$ has the same distribution as $\Pi_{[n]}(X_1, ..., X_n)$, with $X_1, ..., X_n$ defined as in \eqref{eqx} (in particular they are distributed according to the Pitman sampling formula~\eqref{ChEPPF}).

Stated otherwise, we can use the seating plan of $n$ customers to obtain the same partition $\pi_{[n]}^{\text{Db}}$ built through the database (or by partitioning $X_1, ..., X_n$).
Then $\pi^{\text{Db}+}_{[n+1]}$ is obtained when a new customer has chosen an unoccupied table (remember we are in the rare type case), and $\pi^{\text{Db}++}_{[n+2]}$ is obtained when the $n+2$nd customer goes to the same table of the $n+1$st (suspect and crime stain have the same DNA type). In particular, thanks to \eqref{each}, we can write
\begin{equation}\label{e1}
p(\pi_{[n+2]}^{\text{Db++}}\mid H_p, \pi_{[n+1]}^{\text{Db+}}, \alpha, \theta)=1,
\end{equation} 
and 
\begin{equation}\label{e2}
p(\pi_{[n+2]}^{\text{Db++}}\mid H_d, \pi_{[n+1]}^{\text{Db+}}, \alpha, \theta)=\frac{1-\alpha}{n+1+\theta}
\end{equation}

 since the $n+2$nd customer goes to the same table of the $n+1$st where he seats alone.

%Moreover, the map from $Y_1, ..., Y_n$ to $\Pi_{[n]}:=\Pi_{[n]}(Y_1, ..., Y_n)$ is a bijection: from the partition it is possible to go back to the observation $Y_1,...,Y_n$.

%In the forthcoming section we will make use of the following notation:
%$K_n:= |\Pi_{[n]}|$, the number of tables occupied at time $n$.

\section{Some results}\label{gm}
This section presents some useful results that will be used in the forthcoming sections. In particular, Lemma~\ref{lemma1}, suitable to broader applications, is here applied to simplify the LR development. Then, some results from \citet{pitman:2006} regarding the two parameter Poisson Dirichlet distribution, are listed. 
\subsection{A useful Lemma}
The following lemma is a result regarding four general random variables $A$, $X$, $Y$, $H$ whose conditional dependencies are described by the Bayesian network of Figure~\ref{figure1}. The importance of this result is due to the possibility of applying it to a very common forensic situation: the prosecution and the defense disagree on the entirety of data ($Y$) but agree on a part of it ($X$) (indeed, as already noticed, defense and prosecution agree on the distribution of $\pi_{[n+1]}^{\text{Db+}}$, but not on the distribution of $\pi_{[n+2]}^{Db++}$). Data depends on parameters ($A$). 

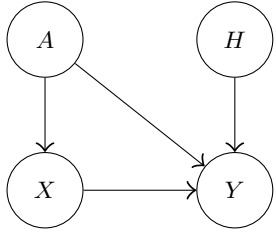
\begin{figure}[htbp]
  \begin{tikzpicture}
\node[draw, circle, minimum size=1cm]                (t) at (-1,0)  { $A$};
\node [draw, circle, minimum size=1cm]              (h) at (1.5,0) {$H$};
 \node [draw, circle, minimum size=1cm]              (d) at (-1,-2) { $X$};
\node [draw, circle, minimum size=1cm]              (dr) at (1.5,-2) { $Y$};
  \draw[black, big arrow]  (t) -- (d);
  \draw[black, big arrow] (h) -- (dr);     
  \draw[black, big arrow]  (t) -- (dr);
   \draw[black, big arrow] (d) -- (dr);
\end{tikzpicture}
\label{figure1}
     \caption{Conditional dependencies of the random variables of Lemma\ref{lemma1}}
     \end{figure}
  
   \begin{lemma}\label{lemma1}
Given four random variables $A$, $H$, $X$ and $Y$, whose conditional dependencies are represented by the Bayesian network of Figure~\ref{figure1}, the likelihood function for $h$, given $X=x$ and $Y=y$ satisfies 
$$\mathrm{lik}(h\mid x, y)~ \propto \mathbb{E}(p(y \mid x, A, h) \mid X = x).$$   \end{lemma}

     \begin{proof}

The model of Figure~\ref{figure1} represents four variables $A$,  $H$, $X$ and $Y$ whose joint probabilty density can be factored as 
$$p(a, h, x, y) ~=~ p(a) \, p(x \mid a)\, p(h)\, p(y \mid x, a, h).$$

\noindent By Bayes formula, $p(a) \, p(x \mid a)= p(x)\,p(a\mid x)$. This rewriting corresponds to reversing the direction of the arrow between $A$ and $X$:
     
\bigskip     

\begin{center}
\begin{tikzpicture}
\node [draw, circle, minimum size=1cm]              (a) at (-1,0)  { $A$};
\node [draw, circle, minimum size=1cm]              (h) at (1.5,0) {$H$};
\node [draw, circle, minimum size=1cm]              (x) at (-1,-2) {$X$};
\node [draw, circle, minimum size=1cm]              (y) at (1.5,-2) { $Y$};

\draw [black, big arrow] (x) -- (a);
\draw [black, big arrow] (a) -- (y);     
\draw [black, big arrow] (h) -- (y);
\draw [black, big arrow] (x) -- (y);
\end{tikzpicture}
\end{center}
\noindent The random variable $X$ is now a root node. This means that when we probabilistically condition on $X=x$, the graphical model changes in a simple way: we can delete the node $X$, but just insert the value $x$ as a parameter in the conditional probability tables of the variables $A$ and $Y$ which formerly had an arrow from node $X$. The next graph represents this model:
\bigskip

\begin{center}
\begin{tikzpicture}
\node [draw, circle, minimum size=1cm]              (a) at (-1,0)  { $A$};
\node [draw, circle, minimum size=1cm]              (h) at (1.5,0) {$H$};
\node                                                                    (x1) at (-1,-1.5) { $x$};
\node                                                                    (x2) at (0,-2) { $x$};
\node [draw, circle, minimum size=1cm]              (y) at (1.5,-2) { $Y$};

\draw [black, dash arrow, dashed] (x1) -- (a);
\draw [black, big arrow] (a) -- (y);     
\draw [black, big arrow] (h) -- (y);
\draw [black, dash arrow, dashed] (x2) -- (y);
\end{tikzpicture}
\end{center}
\noindent This tells us, that conditional on $X=x$, the joint density of $A$, $Y$ and $H$ is equal to 
$$p(a\mid x) p(h) p(y \mid x, a, h).$$
The joint density of $H$ and $Y$ is obtained by integrating out the variable $a$. It can be expressed as a conditional expectation value, since $p(a\mid x)$ is the density of $A$ given $X=x$. We find:
$$p(h) \mathrm{E}(p(y \mid x, A, h) \mid X = x).$$

Recall that this is the joint density of two of our variables, $H$ and $Y$, after conditioning on the value $X=x$. Let us now also condition on $Y=y$. It follows that the density of $H$ given $X=x$ and $Y=y$ is proportional (as function of $H$, for fixed $x$ and $y$) to the same expression, $p(h) \mathrm{E}(p(y \mid x, A, h) \mid X= x)$. 

This is a product of the prior for $h$ with some function of $x$ and $y$. Since posterior odds equals prior odds times likelihood ratio, it follows that the likelihood function for $h$, given $X=x$ and $Y=y$ satisfies
$$\textrm{lik}(h\mid x, y)~ \propto \mathrm{E}(p(y \mid x, A, h) \mid X = x).$$
%     
%The Bayesian network of Figure~\ref{figure1} allows to write the joint density of the four variables as:
%$$p(x, y,a, h)= p(a) p(x | a) p(h) p(y | a, x, h)\quad \forall x,y,a,h.$$
%Because of Bayes' theorem we have that $p(a) p(x | a)=p(x) p(a|x)$, which results in
%$$p(x, y, a, h)= p(x) p(a | x) p(h) p(y | a, x, h)\quad \forall x,y,a,h.$$
%As a consequence:
%$$p(x, y, h)= p(x) p(h) \mathbb{E}(p(y | A, x, h)|\mathcal{X}=x) \quad \forall x,y,a,h$$
%and
%$$p(x, y |h)= p(x) \mathbb{E}(p(y | A, x, h)|\mathcal{X}=x) \quad \forall x,y,a,h$$
%It is then straightforward to check~\eqref{LR_fo}.
%
%
%

\end{proof}

   \begin{corollary}\label{lemma21}
Given four random variables $A$, $H$, $X$ and $Y$, whose conditional dependencies are represented by the network of Figure~\ref{figure1}, the likelihood ratio for $H=h_1$ against $H=h_2$ given $X=x$ and $Y=y$ satisfies 
         \begin{equation}\label{LR_fo}
         \text{LR}=\frac{\mathbb{E}(p(y|x, A, h_1)|X=x)}{\mathbb{E}(p(y|x, A, h_2)|X=x)}.
         \end{equation}
   \end{corollary}  

\subsection{Known results about the two parameter Poisson Dirichlet distribution}
We will now list some theoretical results which will be useful in the forthcoming analysis. Most of these results can be found in \citet{pitman:2006}.

Denote as $K_{n}$ the random number of blocks of a partition $\Pi_{[n]}$ distributed according to the Pitman sampling formula with parameters $\alpha$ and $\theta$.
\begin{itemize}

\item It exists a positive random variable $S_{\alpha}$ such that  
\begin{equation}
\lim_{n\rightarrow +\infty} \frac{K_n}{n^\alpha} = S_{\alpha} \quad \text{a.s.}
\end{equation}
the distribution of $S_{\alpha}$ is a generalization of the Mittag Leffler distribution \citep{ gorenflo:2014}.

\item If $\boldsymbol{P} \sim \text{PD}(\alpha, \theta)$, then
\begin{equation} \label{eqz}
P_i \rightarrow Z i^{-1/\alpha}, \quad \text{a.s., when  } i\rightarrow +\infty
\end{equation} 
for a random variable $Z$ such that $Z^{-\alpha}= \Gamma(1-\alpha)/S_{\alpha}$.

\item For a fixed $\alpha \in (0,1)$, the PD($\alpha, \theta$) (for different $\theta$) are all mutually absolutely continuous. This means that $\theta$ cannot be consistently estimated for $\alpha$ in the range of interest. On the other hand, the power law behavior described above tells us that $\alpha$ can be consistently estimated. 

\item  Studying \eqref{eqch} one can see that when $n$ increases, the parameter $\theta$ becomes less and less important. However, it describes how much ``social'' are the customers: the smaller $\theta$ the more the customers tend to seat to already occupied tables. Thus, it determines the sizes of the big tables, but it won't be much important for our application (the more rare DNA types). 

\item Given $\Pi_{n}$ distributed according to Pitman sampling formula \eqref{ChEPPF}, it holds that 
\begin{equation}
\lim_{n\rightarrow +\infty} \frac{m_j(n)}{n^\alpha}= \frac{\alpha \Gamma (j-\alpha)}{\Gamma(1-\alpha)j!} S_{\alpha} \quad \text{a.s.}\, \forall j
\end{equation}
where $m_j(n)$, $j=1, ..., n$ the random number of blocks of the partition $\Pi_{[n]}$ of size $j$.  This result is presented in \citet{gnedin:2007}, based on \citet{karlin:1967}.

\end{itemize}

\section{The likelihood ratio}\label{lr}

The hypotheses of interest are:
\begin{itemize}
\item $H_p$ = The crime stain was left by the suspect.
\item $H_d$ = The crime stain was left by someone else.
\end{itemize}

The LR will thus be defined as $$\text{LR}= \frac{p(\pi_{[n+2]}^{\text{Db}++}|H_p)}{p(\pi_{[n+2]}^{\text{Db}++}|H_d)}= \frac{p(\pi_{[n+1]}^{\text{Db}+},\pi_{[n+2]}^{\text{Db}++}|H_p)}{p(\pi_{[n+1]}^{\text{Db}+},\pi_{[n+2]}^{\text{Db}++}|H_d)}.$$

where the last equality is due to the fact that $\Pi_{[n+1]}^{\text{Db}+}$ is deterministic, given $\Pi_{[n+2]}^{\text{Db}++}$.

Corollary~\ref{lemma21} with $A= (A,\Theta)$, $X= \Pi_{[n+1]}^{\text{Db}+}$, $Y= \Pi_{[n+2]}^{\text{Db}++}$, and $H=H$ allows to obtain the LR as:
   \[
\begin{aligned}\label{ffd}
\text{LR}~&=\frac{\mathbb{E}(p(\pi_{[n+2]}^{\text{Db}++}\mid \pi_{[n+1]}^{\text{Db}+}, A, \Theta, H_p ) \mid \Pi_{[n+1]}^{\text{Db}+}= \pi_{[n+1]}^{\text{Db}+}) }{ \mathbb{E}(p(\pi_{[n+2]}^{\text{Db}++}\mid \pi_{[n+1]}^{\text{Db}+}, A, \Theta, H_d ) \mid \Pi_{[n+1]}^{\text{Db}+}= \pi_{[n+1]}^{\text{Db}+}) }\\[11pt]
&=\frac{1}{\mathbb{E}\Big(\frac{1-A}{n+1+\Theta}\mid \Pi_{[n+1]}^{\text{Db}+}= \pi_{[n+1]}^{\text{Db}+}\Big)}.
\end{aligned}
\]
where the last equality is due to \eqref{e1} and \eqref{e2}.
By defining the random variable ${\displaystyle\Phi=n\frac{1-A}{n+1+\Theta}}$ we can write the LR as
\begin{equation}\label{ffd2}
\text{LR}=\frac{n}{\mathbb{E}(\Phi \mid \Pi_{[n+1]}=\pi_{[n+1]})}.
\end{equation}

 \subsection{True LR}\label{tlr}
 In order to evaluate the performance of this method one would like to compare the LR values obtained with \eqref{ffd2} with the `true' ones, meaning the LR values obtained when vector $\mathbf{p}$ is known, which means that we have the list of the frequencies of all the DNA types in the population of interest. The LR in this case can be obtained in the following way:
 \begin{align} \label{ddd}
 \text{LR}&=\frac{p(\pi_{[n+2]}^{\text{Db++}}|\pi_{[n+1]}^{\text{Db+}},H_p, \mathbf{p})}{p(\pi_{[n+2]}^{\text{Db++}}|\pi_{[n+1]}^{\text{Db+}},H_d, \mathbf{p})}=\frac{1}{p(\pi_{[n+2]}^{\text{Db++}}|\pi_{[n+1]}^{\text{Db+}}, H_d, \mathbf{p})}\\
&= \frac{1}{\Pr(X_{n+2}=X_{n+1}|\pi_{[n+1]}^{\text{Db+}}, H_d, \mathbf{p})}\\
&=\frac{1}{{\displaystyle \sum_{(x_1, ..., x_{n+1})} \Pr(X_{n+2}=x_{n+1}|x_1, ..., x_{n+1},\pi_{[n+1]}^{\text{Db+}}, H_d, \mathbf{p} })p(x_1, ..., x_{n+1}|\pi_{[n+1]}^{\text{Db+}}, \mathbf{p})}\\
&=\frac{1}{{\displaystyle \sum_{(x_1, ..., x_{n+1})}} p_{x_{n+1}} p(x_1, ..., x_{n+1}|\pi_{[n+1]}^{\text{Db+}}, \mathbf{p})}\\
&= \frac{1}{\mathbb{E}(p_{x_{n+1}}|\pi_{[n+1]}^{\text{Db+}}, \mathbf{p})}. \label{ddd2}
 \end{align}

Notice that in the rare type case $x_{n+1}$ is observed only once among the $x_1$, ..., $x_{n+1}$. Hence we call it a singleton. Let $N_1$ denote the number of singletons, and $\mathcal{S}$ the set of all singletons. 
Given $\mathbf{p}$ and $\pi_{[n+1]}^{\text{Db+}}$, it holds that the distribution of $X_{n+1}$ is the same as the distribution of all other $N_1$ singletons. This implies that:

$$
N_1 \mathbb{E}(p_{x_{n+1}}|\pi_{[n+1]}^{\text{Db+}}, \mathbf{p}) =  \mathbb{E}(\sum_{x_i \in \mathcal{S}}p_{x_{i}}|\pi_{[n+1]}^{\text{Db+}}, \mathbf{p}).
$$

Let us denote as $X_1^*$, .., $X^*_K$ the $K$ different values taken by $X_1$, ..., $X_{n+1}$,  ordered according to the frequency of their values. Stated otherwise, if $n_i$ is the frequency of $x_i^*$ among $x_1, ...,  x_{n+1},$ then $n_1\geq n_2\geq ... \geq n_K$. Moreover, in case $X_i^*$ and $X_j^*$ have the same frequency ($n_i=n_j$),  than they are ordered according to their values. For instance, if $X_1= 4$,  $X_2=2$, $X_3=4$, $X_4=2$, $X_5= 3$, $X_6=3$, $X_7=10$, $X_8=13$, $X_9=5$, $X_{10}=2$, $X_{11}=9$, then 
$X_1^*= 2, X_2^*=3, X_3^*=4, X_4^*=5, X_5^*=9, X_6^*=10, X_7^*=13$.

By definition, it holds that
$$\mathbb{E}(\sum_{x_i \in \mathcal{S}}p_{x_{i}}|\pi_{[n+1]}^{\text{Db+}}, \mathbf{p}) = \mathbb{E}(\sum_{j: \, n_j=1}p_{x_{j}^*}|\pi_{[n+1]}^{\text{Db+}}, \mathbf{p}).
$$

Notice that $(n_1, n_2, ..., n_K)$ is a partition of $n+1$, which will be denoted as $\pi_{n+1}^{\text{Db+}}$. In the example, $\pi_{n+1}^{\text{Db+}}=(3, 2, 2, 1, 1, 1, 1)$. A more compact representation for $\pi_{n+1}^{\text{Db+}}$ can be obtained by using two vectors $\mathbf{a}$ and $\mathbf{r}$ where $a_j$ are the
distinct numbers occurring in the partition, ordered, and each $r_j$ is the number of
repetitions of $a_j$. $J$ is the length of these two vectors, and it holds that $n+1=\sum_{j=1}^J a_j r_j.$
In the example above we have that $\pi_{n+1}^{\text{Db+}}=(\mathbf{a}, \mathbf{r})$ with $\mathbf{a}=(1,2,3)$ and $\mathbf{r}=(4, 2,1)$.

Since the distribution of ${\displaystyle \sum_{j: \, n_j=1}p_{x_{j}^*}}$ only depends on $\pi_{n+1}^{\text{Db+}}$, the latter can replace $\pi_{[n+1]}^{\text{Db+}}$. Thus, it holds that
\begin{equation}\label{lrp}
\textrm{LR}=\frac{N_1}{\mathbb{E}({\displaystyle \sum_{j: \, n_j=1}}p_{x_{j}^*}|\pi_{n+1}^{\text{Db+}}, \mathbf{p})}.
\end{equation}

Notice that the knowledge of $\mathbf{p}$, is not enough to observe $X_1^*, ..., X_K^*$: $\mathbf{p}$ is sorted in decreasing order, and even if we know the different values $p_i$, we don't know to which category each value belongs.
There is a function, $\chi$, treated here as latent variable, which assigns all DNA types, ordered according to their frequency in Nature, to one of the number $ \{1, 2, ..., J\}$ corresponding to the position in $\mathbf{a}$ of its frequency in the sample, or to $0$ if the type if not observed. Stated otherwise, $$ \chi: \{1, 2, ...\} \longrightarrow \{1, 2, ..., J\}.$$
$$\chi(i)=\begin{cases}
0 & \textrm{if the $i$th most common species in Nature is not observed in the sample},\\
j & \textrm{if the $i$th most common species in Nature is one of the $r_j$ observed $a_j$ times in the sample.}
\end{cases}$$
 
 Given $\pi_{n+1}^{\text{Db+}}=(\mathbf{a}, \mathbf{r})$, $\chi$ must satisfy the following conditions:
 \begin{equation}\label{condfi}
 \sum_{i=1}^{\infty} \mathbf{1}_{\chi(i)=j}=r_j, \qquad \forall j.
 \end{equation}

The map $\chi$ can be represented by a vector $\boldsymbol{\chi}=(\chi_1, \chi_2, ...)$ made of its values: $\chi_i=\chi(i)$. 
In the example above we have that $\chi=(0, 3, 2, 2, 1, 0, 0, 0, 1, 1, 0, 0, 1,0 ...0)$.

Notice that, given $\pi_{n+1}^{\text{Db+}}=(\mathbf{a},\mathbf{r})$, the knowledge of $\boldsymbol{\chi}$ implies the knowledge of $X_1^*$, ..., $X_K^*$: indeed it is enough to sort the positive values among the $\chi_i$ and take their positions in $\chi$ solving ties by considering the position themselves (if $\chi_i=\chi_j$, than the order is given by $i$ and $j$). 
For instance, in the example, if we sort the values of $\chi$ and we collect their positions we get
$(2, 3, 4, 5, 9, 10, 13)$: the reader can notice that we got back to the $X_1^*, ..., X_7^*$.

This means that to obtain the distribution of $X_1^*,..., X_{K}^*| \pi_{n+1}^{\text{Db+}}, \mathbf{p}$, which appears in~\eqref{lrp}, it is enough to obtain the distribution of $\boldsymbol{\chi}|\pi_{n+1}^{\text{Db+}}, \mathbf{p}$. Actually, we are only interested in the mean of the sum of singletons in samples of size $n+1$ from the distribution of $X_1^*,..., X_{K}^*| \pi_{n+1}^{\text{Db+}}, \mathbf{p}$: this means that we can just simulate samples from the distribution of $\boldsymbol{\chi}|\pi_{n+1}, \mathbf{p}$ and sum those $p_a$ such that $\chi_a=1.$

To simulate samples from the distribution of $\boldsymbol{\chi}|\pi_{n+1}, \mathbf{p}$ we use a Metropolis - Hasting algorithm, on the space of the vectors satisfying condition~\eqref{condfi}.
Notice that for the model we assumed $\mathbf{p}$ to be infinitely long, but for simulations we will use a finite $\mathbf{\bar{p}}$, of length $M$. This is equivalent to assume that only $M$ elements in the infinite $\mathbf{p}$ are positive, and the remaining infinite tail is made of zeros. 
Then the state space of the Metropolis Hasting Markov chain is made of all vectors of length $M$ whose elements belong to $\{0, 1, ..., J\}$, and satisfy the condition~\eqref{condfi}. 
If we start with a initial point $\boldsymbol{\chi}_{0}$ which satisfies~\eqref{condfi} and, at each allowed move of the Metropolish Hasthing , we swap two different values $\chi_a$ and $\chi_b$ inside the vector, condition \eqref{condfi} remains satisfied. The algorithm is based on a similar one proposed in~\citet{anevski:2013}. 

This method allows us to obtain the `true' LR when the vector $\mathbf{p}$ is known. This is rarely the case, but we can put ourselves in a fictitious world where we know $\mathbf{p}$, and compare the true values for the LR with the one obtained by applying our model when $\mathbf{p}$ is unknown. This will be done in the forthcoming section.

\section{Analysis on a real database}\label{ard}
In this section we present the study we made on a database of 18,925 Y-STR 23-loci profiles from 129 different locations in 51 countries in Europe \citep{purps:2014} \footnote{The database has previously been cleaned by Mikkel Meyer Andersen (\url{http://people.math.aau.dk/~mikl/?p=y23}).}.
Different analyses are performed by considering only 7 Y-STR loci (DYS19, DYS389 I, DYS389 II, DYS3904, DYS3915, DY3926,DY3937) but similar results have been observed with the use of 10 loci.

First the maximum likelihood estimators $\alpha_{\text{MLE}}$ and $\theta_{\text{MLE}}$ using the entire database are obtained. Their values are $\alpha_{\text{MLE}}=0.5$ and $\theta_{\text{MLE}}=216.$

In order to check if the choice of the two parameter Poisson Dirichlet prior is a sensible one we first compare the ranked frequencies from the database with the relative frequencies of several samples of size $n$ obtained from realisations of PD($\alpha_{MLE}, \theta_{MLE}$). The asymptotic behaviour described in \eqref{eqz} is also checked.
Lastly, we will analyse the loglikelihood function for data $\pi_{[n+1]}^{\text{Db}+}$ in order to study the denominator of the LR.

\subsection{Model fitting}
\begin{figure}[htbp]%
\includegraphics[scale=0.5]{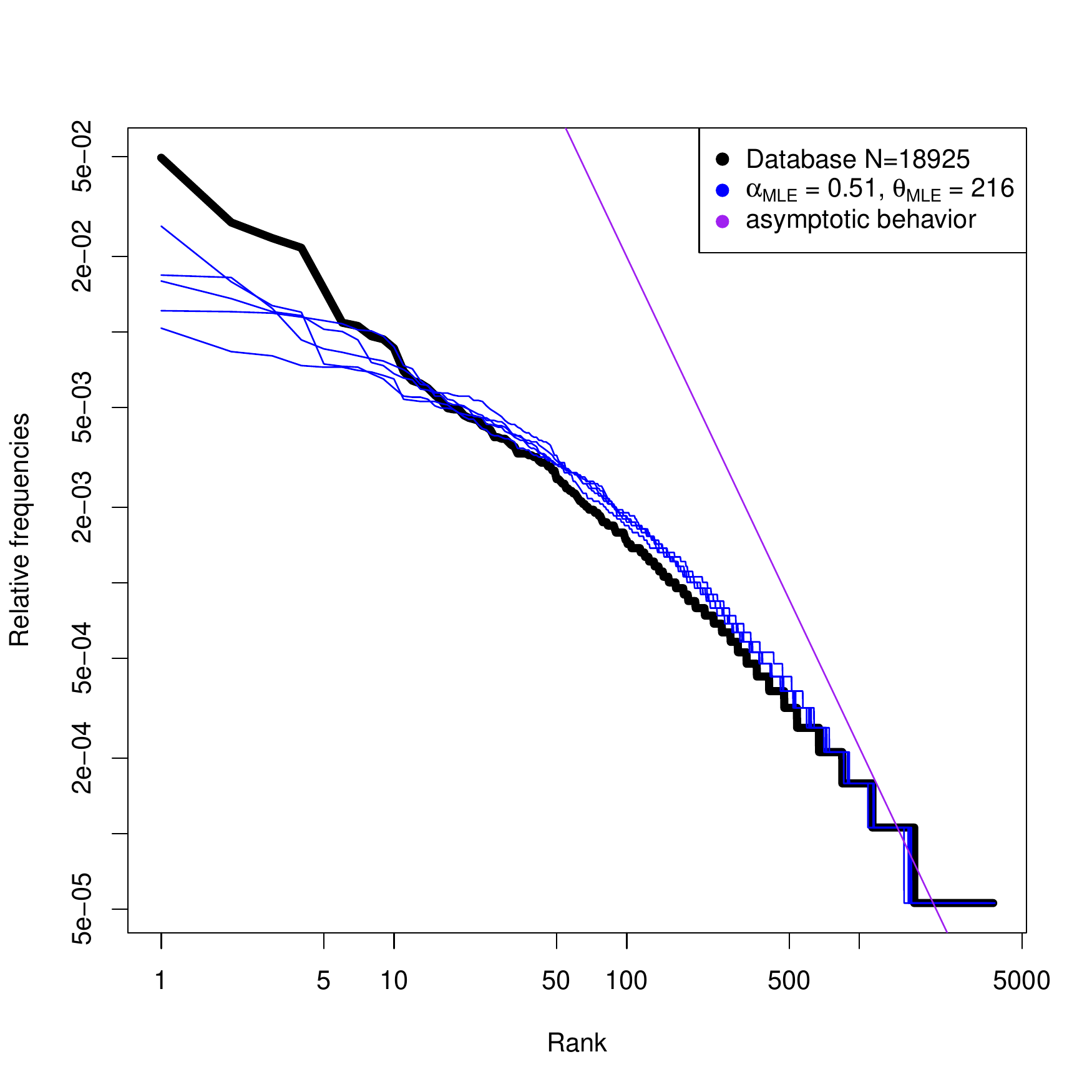}
\caption{Log scale ranked frequencies from the database (thick line) are compared to the relative frequencies of samples of size $n$ obtained from realization of PD($\alpha_{MLE}, \theta_{MLE}$) (thin lines). Asymptotic power law behavior is also displayed (dotted lines).}\label{3figs}
\end{figure}

In Figure~\ref{3figs}, the ranked frequencies from the database are compared to the relative frequencies of samples of size $n$ obtained from several realizations of PD($\alpha_{MLE}, \theta_{MLE}$).
To do so we run several times the Chinese Restaurant seating plan (up to $n=18,925$ customers): each run is equivalent to generate a new realization $\boldsymbol{p}$ from the PD($\alpha_{MLE}, \theta_{MLE}$). The partition of the customers into tables is the same as the partition obtained from an i.i.d. sample of size $n$ from $\boldsymbol{p}$. The ranked relative sizes of each table (thin lines)
are compared to the ranked frequencies of our database (thick line).

\subsection{Asymptotic power law behavior}
The asymptotic behavior described in \eqref{eqz} is also analyzed in Figure~\ref{3figs}.
This behavior is expected in the tail (the limit is over $i$) and clearly the number of customers ($n=18,925$) is not big enough for the small $P_i$ to follow the power law. 

\subsection{Loglikelihood}
It is also interesting to investigate the shape of the loglikelihood function for $\alpha$ and $\theta$ given $\pi_{[n+1]}^{\text{Db}++}$. It is defined as 

$$l_{n+1}(\alpha, \theta):=\log p(\pi_{[n+1]}^{\text{Db}++}| \alpha, \theta).$$
%=\frac{[\theta+ \alpha]_{k ; \alpha}}{[\theta+ 1]_{n; 1}}\prod_{i=1}^{k+1}[1-\alpha]_{n_i-1;1}.$$

In Figure~\ref{2figs} (a), the loglikelihood function is compared to the Gaussian distribution centered in the maximum likelihood estimates for $\alpha$ and $\theta$, with the observed Fisher information as covariance matrix. 
In Figure~\ref{2figs} (b) the loglikelihood reparametrized using $\phi={\displaystyle n\frac{1-\alpha}{n+1+\theta}}$, and $\theta$ instead of $\alpha$ and $\theta$, is displayed and compared to the corresponding Gaussian distribution.

\begin{figure}[htbp]

\centering
\subfigure[Relative loglikelihood for parameters $\alpha$ and $\theta$, compared to a Gaussian distribution, $95\%$ and $99\%$ confidence intervals (green and red) ]{
\includegraphics[scale=0.5,width=0.45\textwidth ]{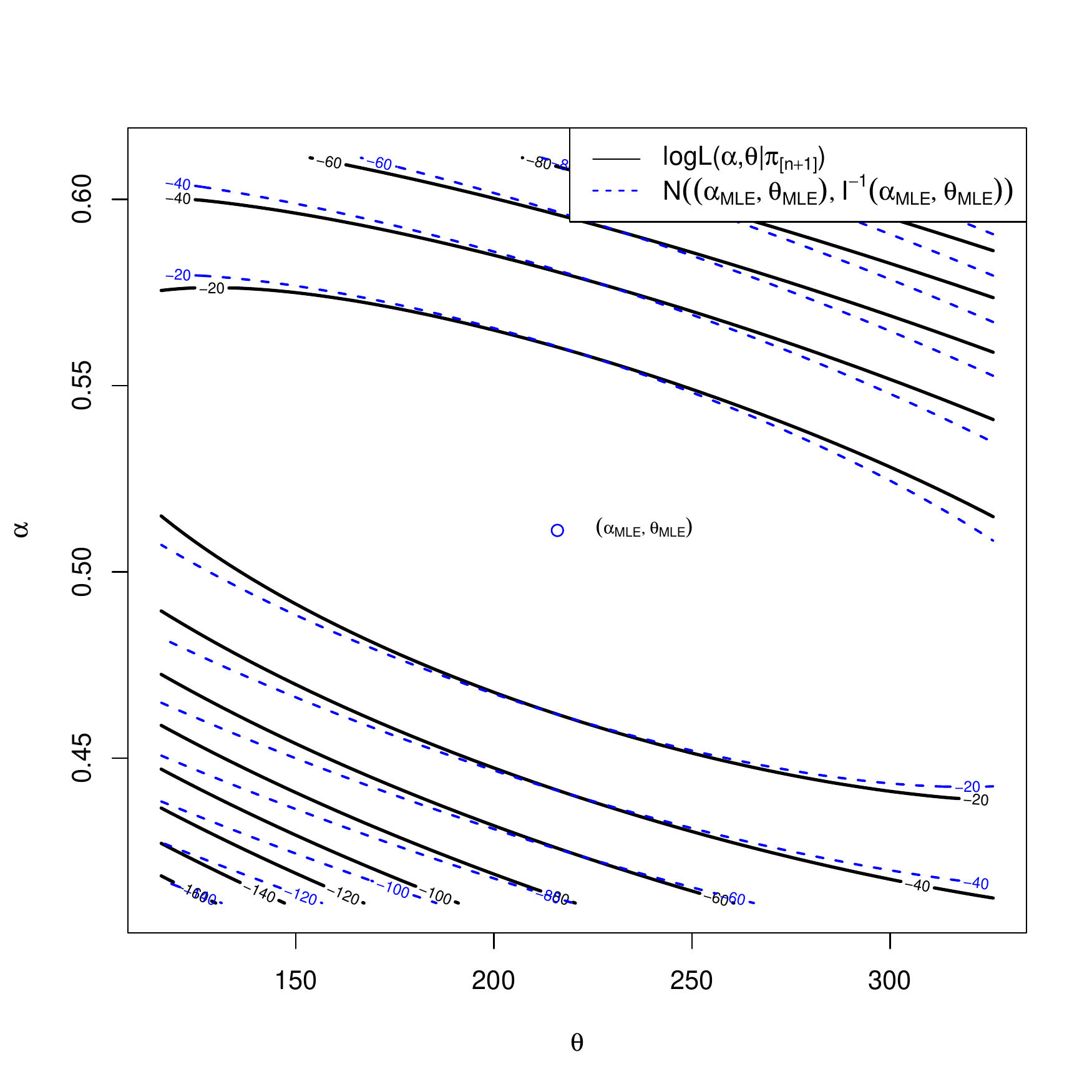}
}
\subfigure[Relative loglikelihood for $\phi=n\frac{1-\alpha}{ n+1+\theta}$ and $\theta$ compared to a Gaussian distribution $95\%$ and $99\%$ confidence intervals (green and red).]{
\includegraphics[scale=0.5,width=0.45\textwidth]{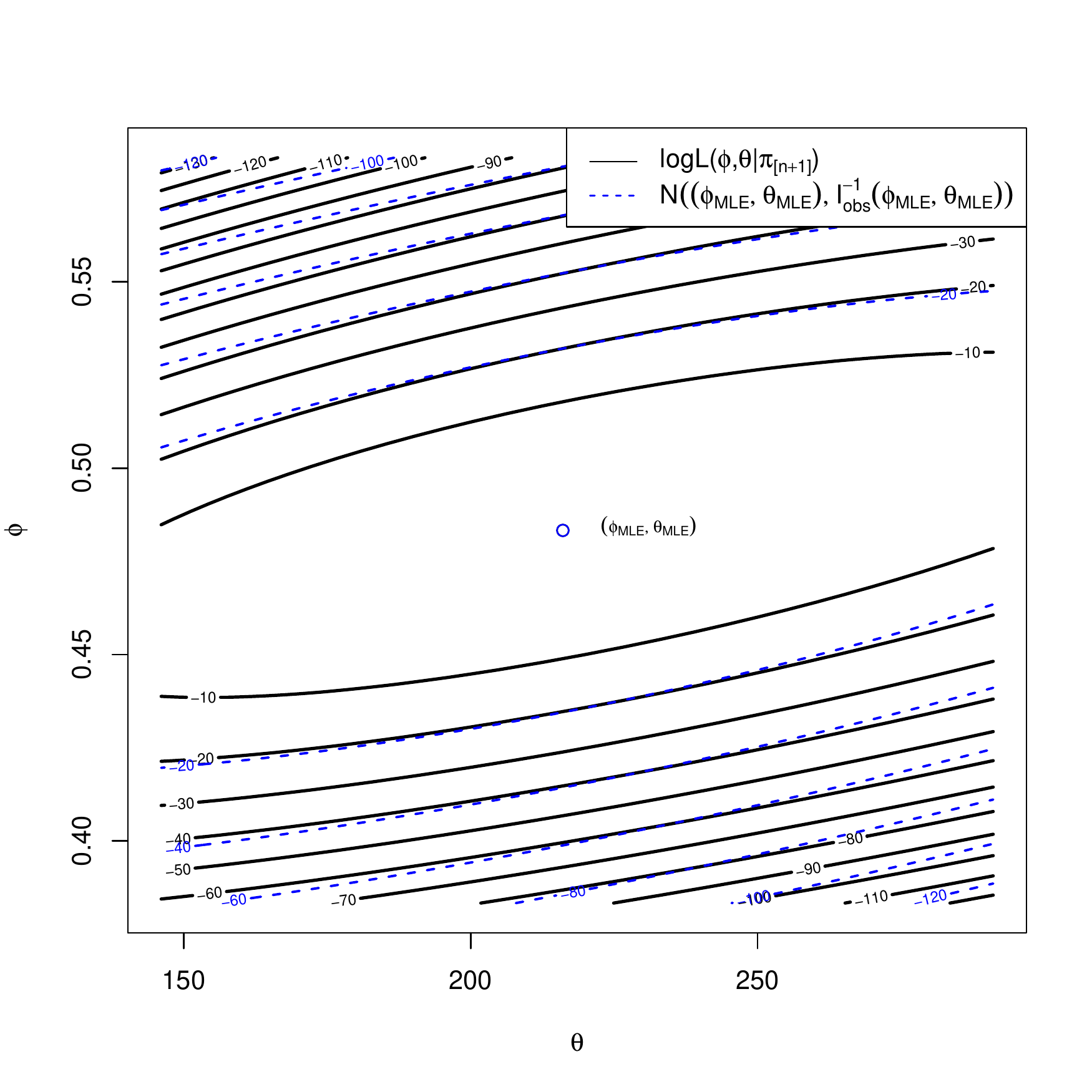}
}
\caption{}\label{2figs}
\end{figure}

The posterior distribution for ($\Phi$,$\Theta$) given $\Pi_{[n+1]}^\text{Db+}$ is proportional to the loglikelihood $l_{[n+1]}(\phi, \theta)$ times the prior $p(\phi, \theta)$. 
The Gaussian behavior of $l_{[n+1]}(\phi, \theta)$ is particularly interesting since it allows to conclude that if the prior $p(\phi, \theta)$ is smooth around $(\phi_{MLE}, \theta_{MLE})$, we can approximate $\mathbb{E}(\Phi|\Pi_{[n+1]}^{\text{Db}++})$ with $\phi_{MLE}={\displaystyle n\frac{1-\alpha_{MLE}}{n+1+\theta_{MLE}}}$. Hence, one can approximate the LR itself in the following way:
\begin{equation} \label{mlelr}
\text{LR}\approx \frac{n+1+\theta_{MLE}}{1-\alpha_{MLE}}.
\end{equation}

Notice that this is equivalent to an hybrid approach, in which the parameters are estimated through the MLE (frequentist) and their value plugged into the Bayesian LR.

\subsection{Analyzing the error}
As explained in Section~\ref{tlr}, a Metropolis Hasting algorithm, based on \citet{anevski:2013}, can be used to obtain the `true LR', that is the LR when the vector $\mathbf{p}$ is known, as defined in \eqref{ddd} - \eqref{ddd2}. The latter will be denoted as $\LR_{|\mathbf{p}}$.
This can be compared to the LR obtained with the method described in this paper, when $\mathbf{p}$ is unknown, as defined in \eqref{ffd} - \eqref{ffd2}. 
Notice that errors appear at different levels \citep{cereda:2015b}: error due to limitedness of samples, error in the model, error  in the choice of the parameters of the model. The following three tests will explore these levels.

\subsubsection*{Test 1}
Instead of using the big database of \citet{purps:2014} we consider its restriction to the Dutch population 
(of size 2085): we pretend this to be the entire population of possible perpetrators, and compare the distribution of $\log_{10}(\LR_{|\mathbf{p}})$ and $\log_{10}\LR$ obtained by 100 samples of size 100 from this population.

\subsubsection*{Test 2}
In order to avoid error due to model selection, we use as entire population a sample from a realization from PD($\alpha, \theta$) distribution. This is done by running the two parameter Chinese restaurant process up to 2085 customers. Then, again, 100 samples of size 100 from this population are sampled and $\log_{10}(\LR_{|\mathbf{p}})$ and $\log_{10}\LR$ are compared. This procedure is repeated 5 times to use 5 different Poisson Dirichlet populations. 

\subsubsection*{Test 3}
The same as in Test 2, but in order to avoid also error due to MLE parameter estimations, we use as parameters $\alpha$ and $\theta$ for $\log_{10}\LR$ exactly those which have been used to generate the Chinese restaurant process. 

\begin{figure}[htbp]

\medskip
\hspace{0.35\baselineskip}\hfil
\makebox[0.4\textwidth]{(a) Comparison}\hfil
\makebox[0.4\textwidth]{(b) Error}\hfil

\settoheight{\tempdim}{\includegraphics[width=0.3\textwidth]{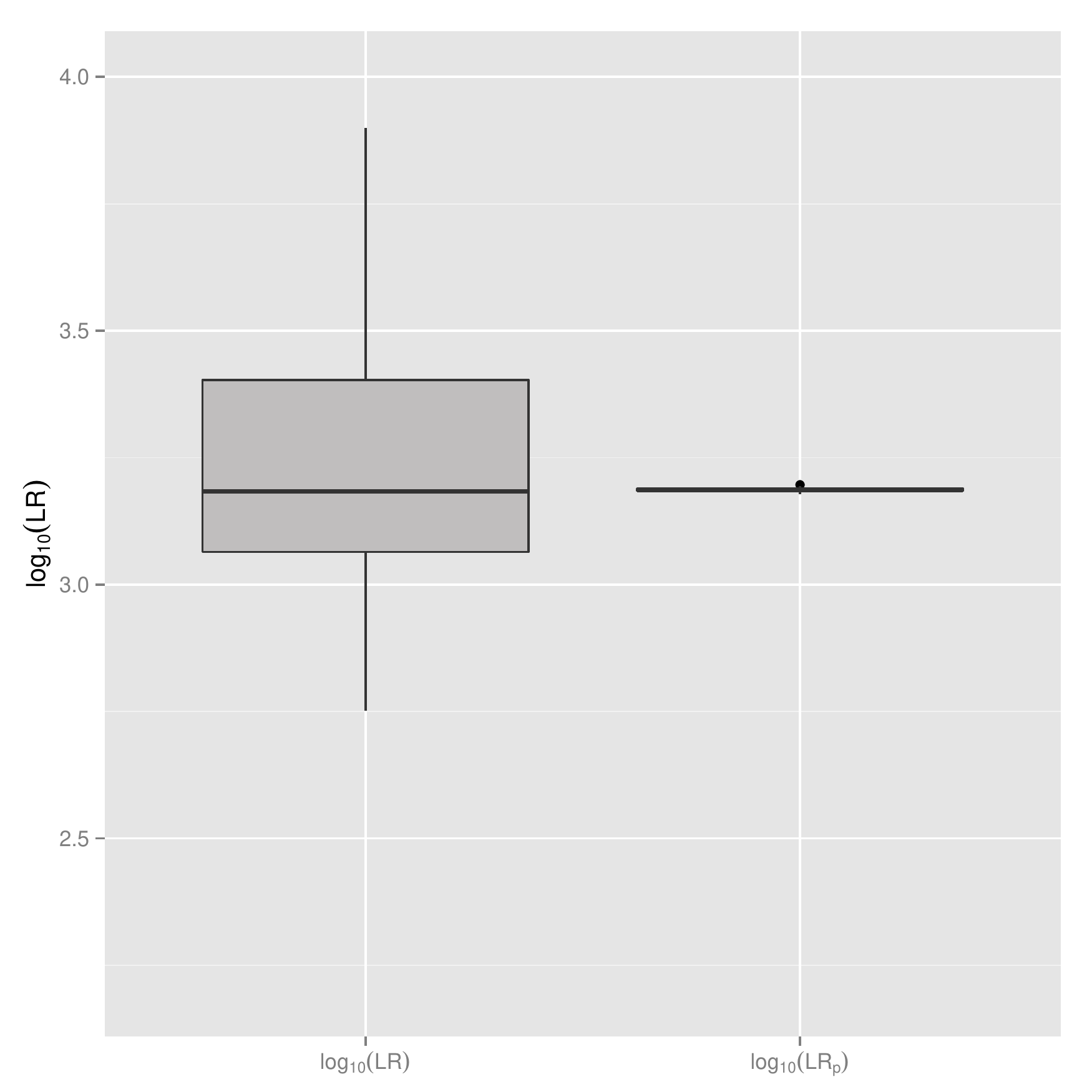}}%
\rotatebox{90}{\makebox[\tempdim]{Test 1}}\hfil
\includegraphics[width=0.4\textwidth]{test1.pdf}\hfil
\includegraphics[width=0.4\textwidth]{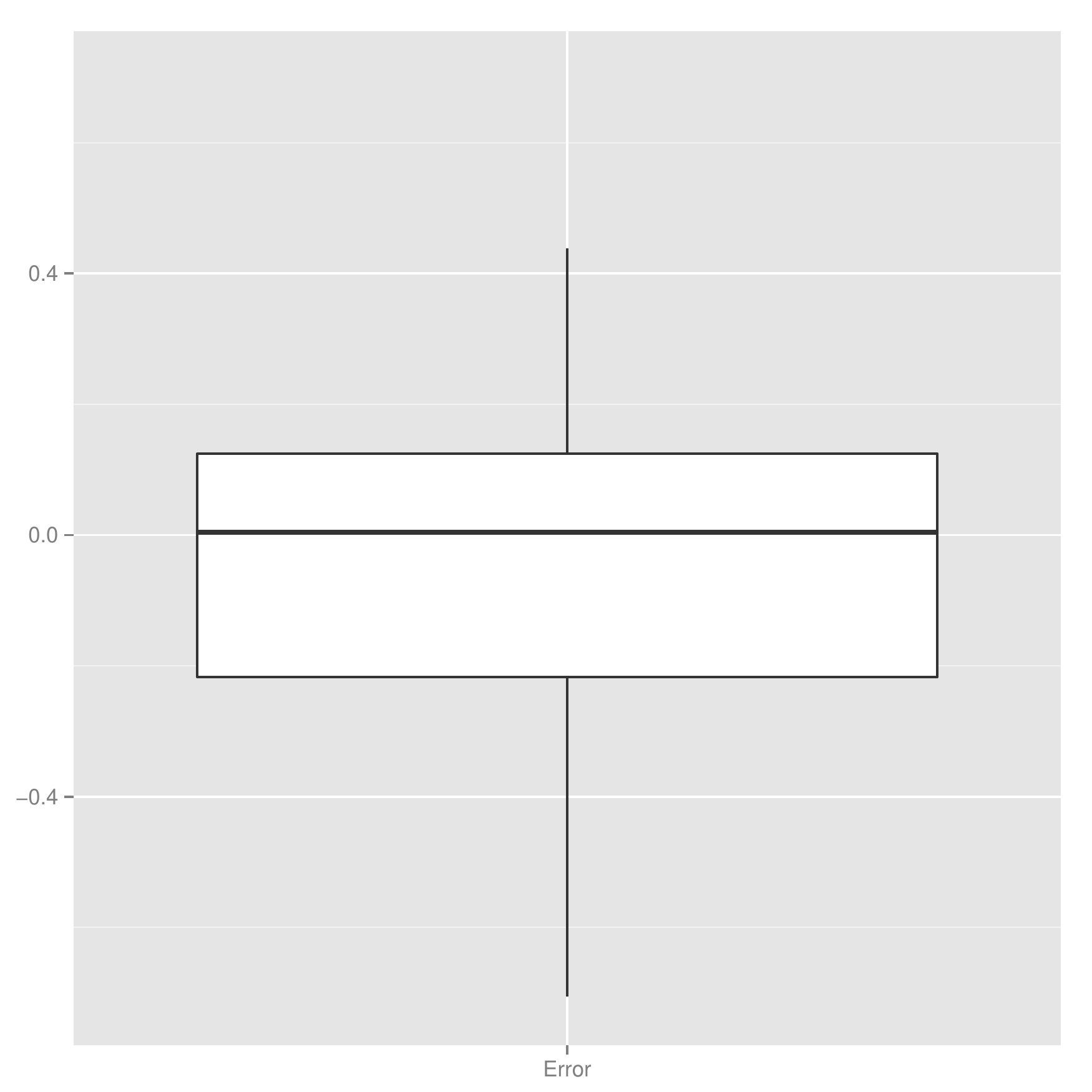}\hfil
%\includegraphics[trim=0cm 0cm 0.5cm 1.5cm, clip=true,width=0.3\textwidth]{big_bn_1.pdf}
%\settoheight{\tempdim}{\includegraphics[width=0.3\textwidth]{bnet.pdf}}%
%\rotatebox{90}{\makebox[\tempdim]{R=10}}\hfil
%\includegraphics[trim=0cm 0cm 0cm 1.5cm, clip=true,width=0.3\textwidth]{LR_bn_10.pdf}\hfil
%\includegraphics[trim=0cm 0cm 0cm 1.5cm, clip=true,width=0.3\textwidth]{diff_bn_10.pdf}\hfil
%\includegraphics[trim=0cm 0cm 0cm 1.5cm, clip=true,width=0.3\textwidth]{big_bn_10.pdf}

\settoheight{\tempdim}{\includegraphics[width=0.3\textwidth]{test1.pdf}}%
\rotatebox{90}{\makebox[\tempdim]{Test 2}}\hfil
\includegraphics[width=0.4\textwidth]{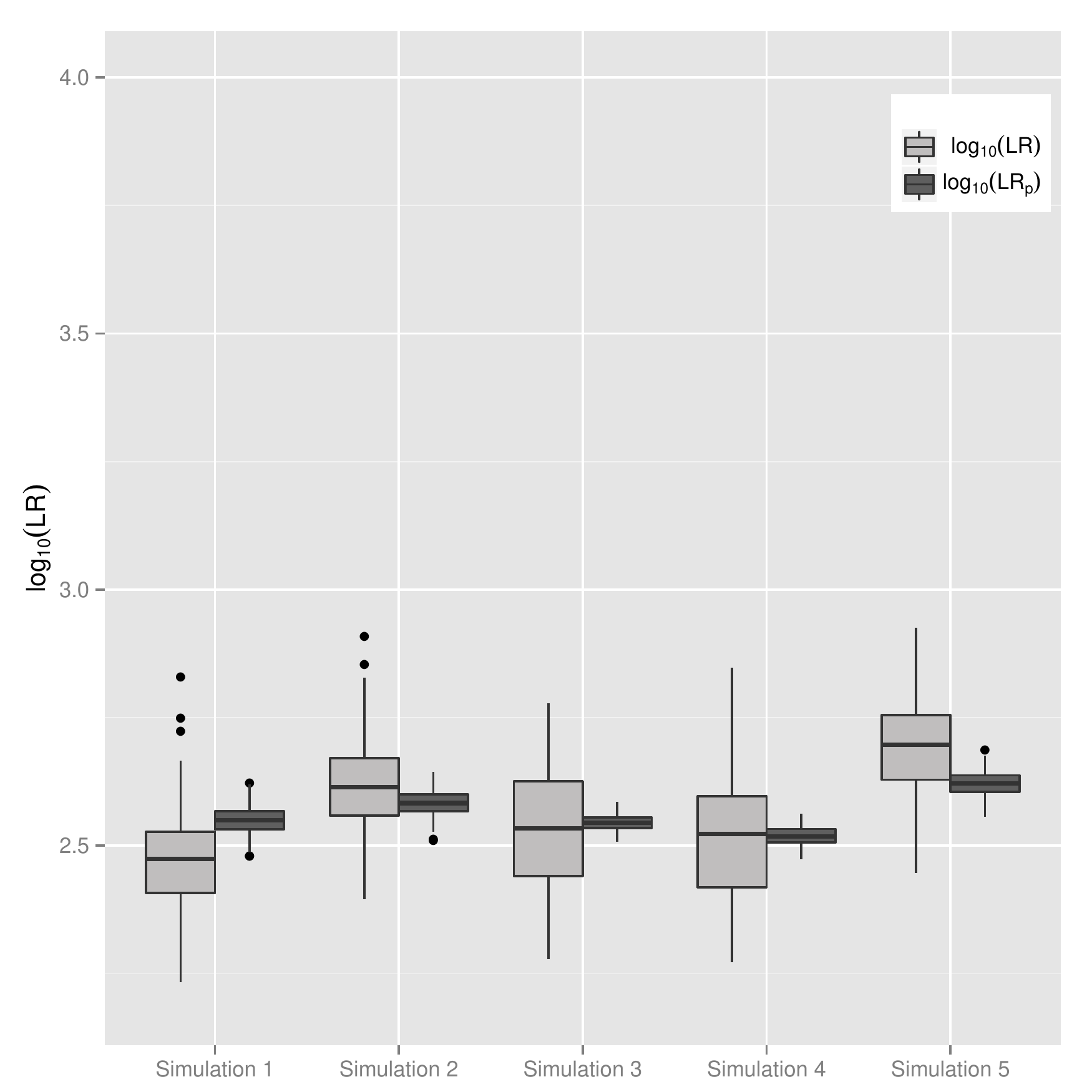}\hfil
\includegraphics[width=0.4\textwidth]{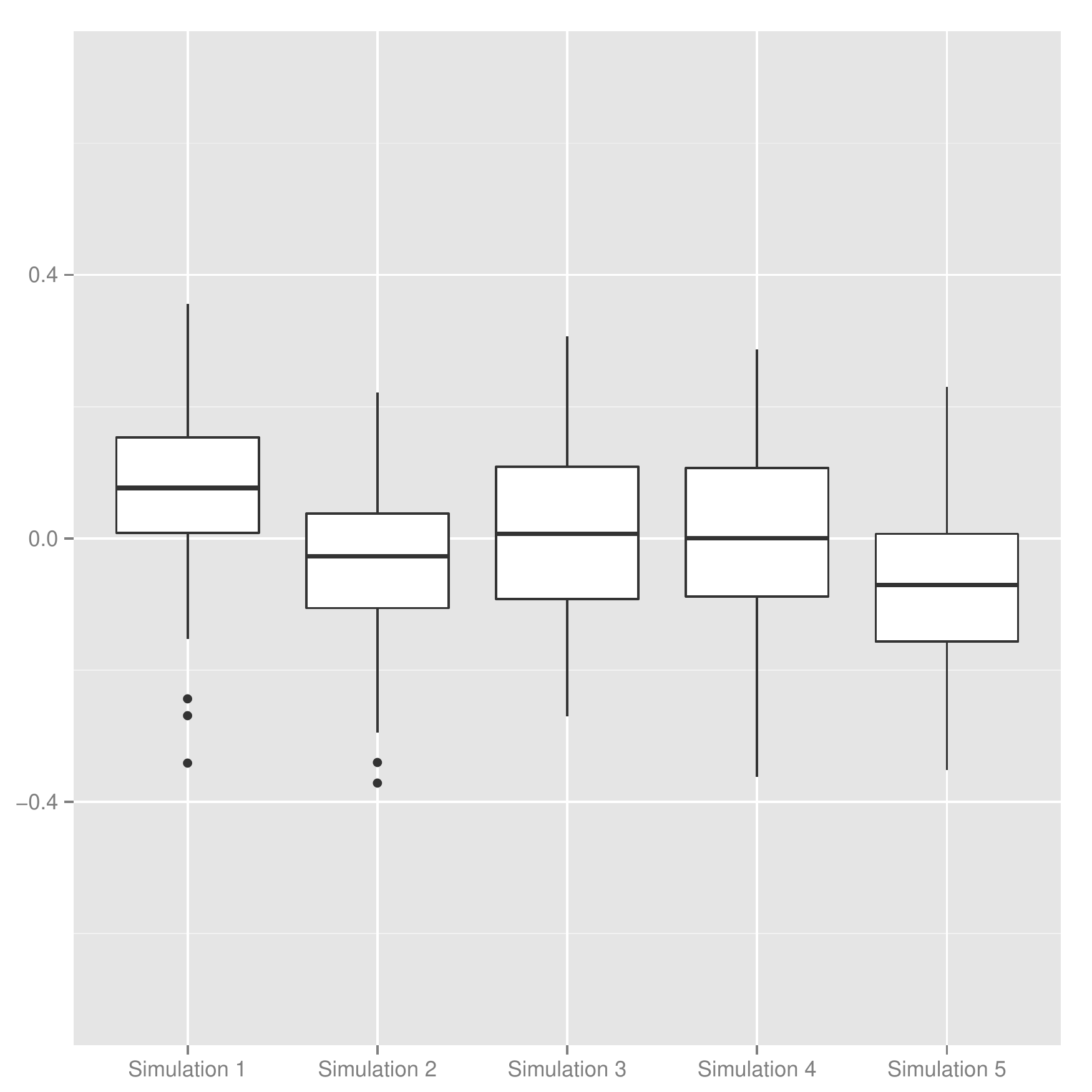}
%\includegraphics[trim=0cm 0cm 0.5cm 1.5cm, clip=true,width=0.3\textwidth]{big_bn_100.pdf}

%\settoheight{\tempdim}{\includegraphics[width=0.3\textwidth]{bnet.pdf}}%
%\rotatebox{90}{\makebox[\tempdim]{R=1000}}\hfil
%\includegraphics[trim=0cm 0cm 0cm 1.5cm, clip=true,width=0.3\textwidth]{LR_bn_1000.pdf}\hfil
%\includegraphics[trim=0cm 0cm 0cm 1.5cm, clip=true,width=0.3\textwidth]{diff_bn_1000.pdf}\hfil
%\includegraphics[trim=0cm 0cm 0cm 1.5cm, clip=true,width=0.3\textwidth]{big_bn_1000.pdf}

\settoheight{\tempdim}{\includegraphics[width=0.3\textwidth]{test1.pdf}}%
\rotatebox{90}{\makebox[\tempdim]{Test 3}}\hfil
\includegraphics[width=0.4\textwidth]{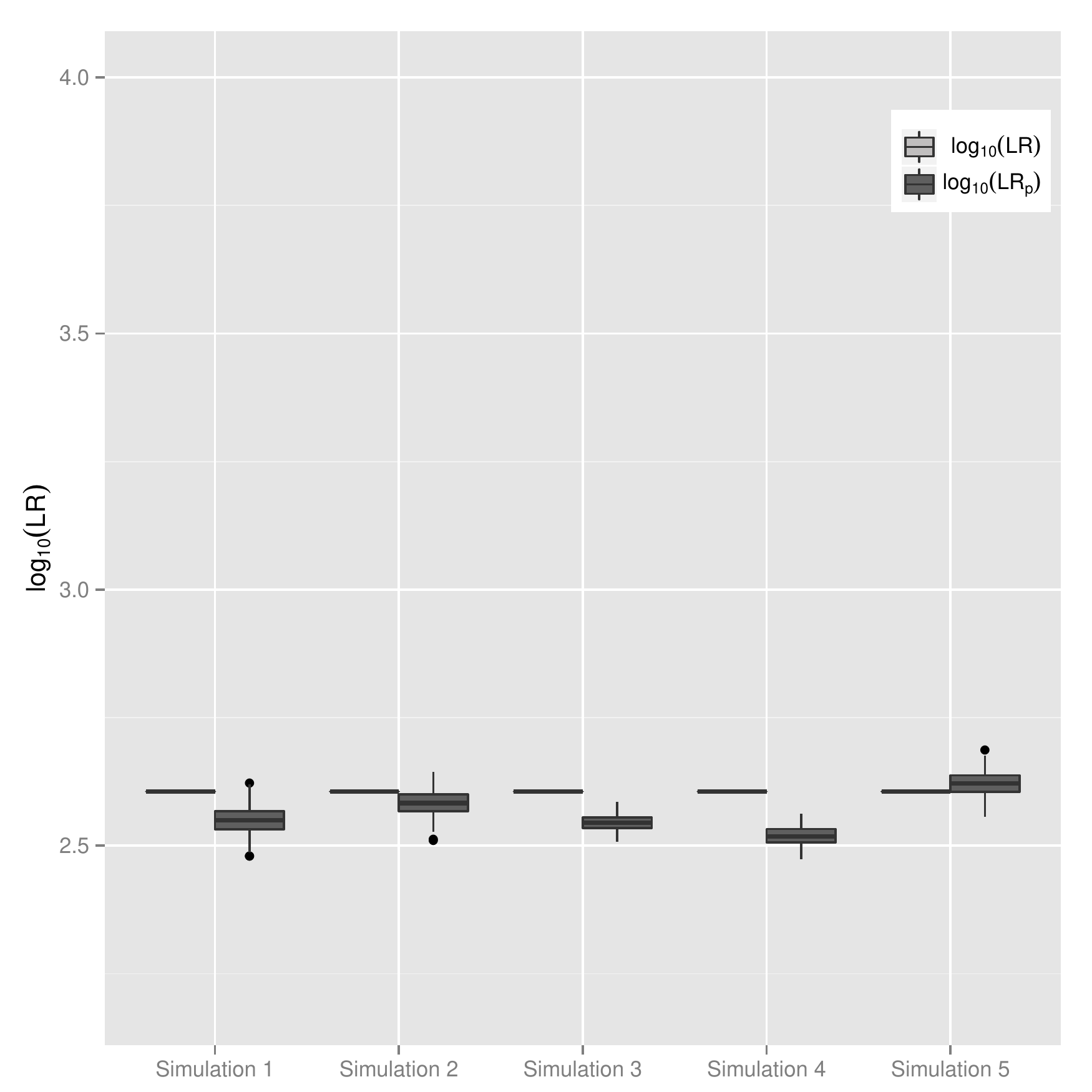}\hfil
\includegraphics[width=0.4\textwidth]{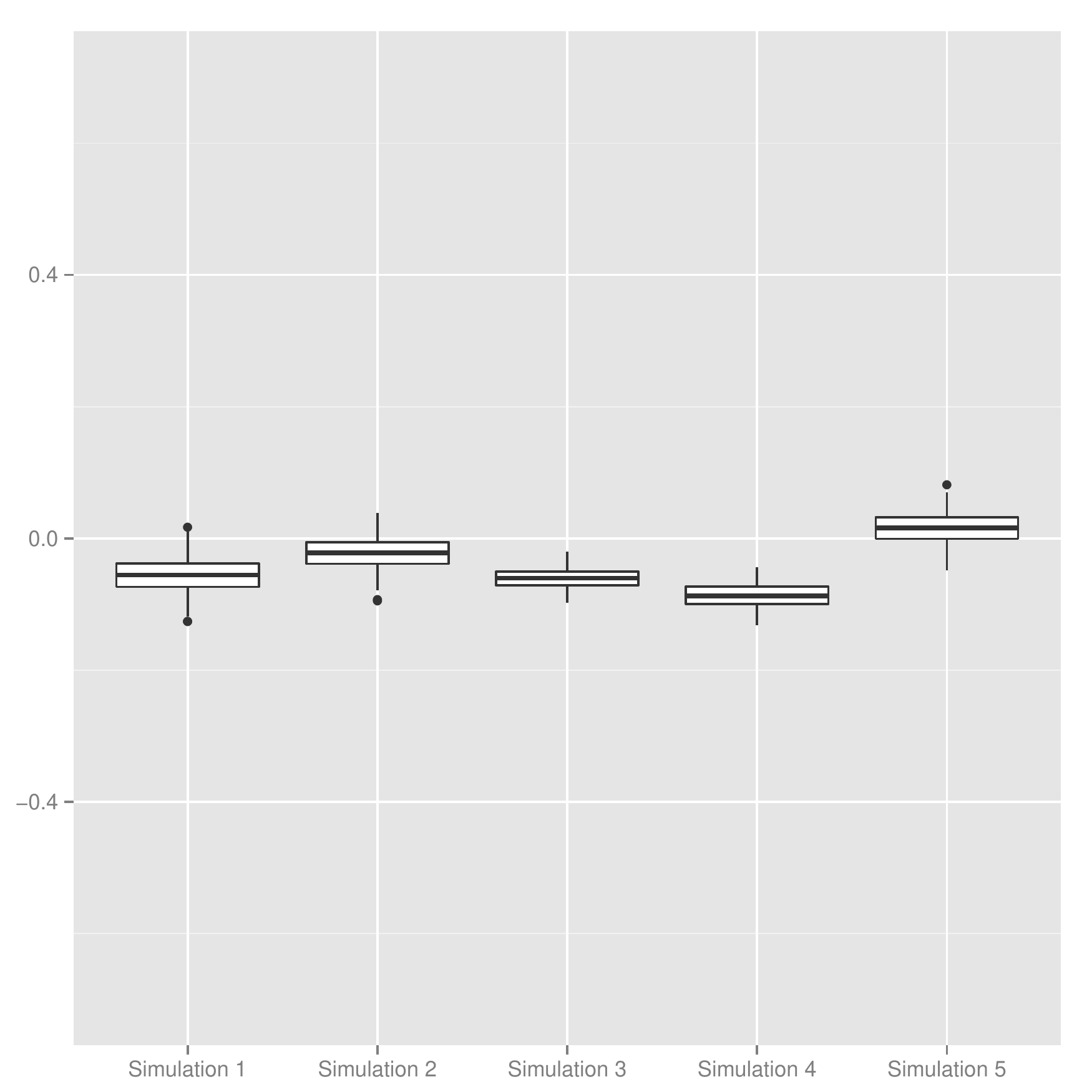}\hfil
%\includegraphics[trim=0cm 0cm 0.5cm 1.5cm, clip=true,width=0.3\textwidth]{big_bn_10000.pdf}

%\settoheight{\tempdim}{\includegraphics[width=0.3\textwidth]{bnet.pdf}}%
%\rotatebox{90}{\makebox[\tempdim]{R=100000}}\hfil
%\includegraphics[trim=0cm 0cm 0cm 1.5cm, clip=true,width=0.3\textwidth]{LR_bn_1_s100000b.pdf}\hfil
%\includegraphics[trim=0cm 0cm 0cm 1.5cm, clip=true,width=0.3\textwidth]{diff_bn_1_s100000b.pdf}\hfil
%\includegraphics[trim=0cm 0cm 0cm 1.5cm, clip=true,width=0.3\textwidth]{big_bn_1_s100000.pdf}

\caption{For Test 1, Test 2 and Test 3 we plot: (a) comparison between the distribution of $\log_{10}LR$  when the population proportions $\mathbf{p}$ are known (LR$_{|\mathbf{p}}$), and when they are not (LR). (b) the error $\log_{10}\LR_{|\mathbf{p}}-\log_{10}LR $.}\label{figggag3}
\end{figure}

\section{Future research questions}~\label{frq}

Because of the mutual absolute continuity results we know that $\theta$ cannot be consistently estimated.
However, there exists at least one consistent estimator for $\alpha$ \citep{carlton:1999}, namely:
$$\hat{\alpha}= \frac{\log K_n }{\log n}.$$
Moreover, $\alpha$ can be estimated consistently also from the power law of \eqref{eqz}.
We are interested in the consistency of $\alpha_{MLE}$, at least when $\theta$ is known, although literature \citep{carlton:1999} is quite skeptical about its  performance.
The observed Fisher information for $\alpha$ grows with $n$ and this gives some hope for the consistency of $\alpha_{MLE}$.

The Gaussian behavior of Figure~\ref{2figs} was unexpected. At least, we expect that increasing $n$, $\alpha$ and $\theta$ would become independent, thus the ellipses will rotate.

\section*{Acknowledgment}
I am indebted to Jim Pitman and Alexander Gnedin for their help in understanding their important theoretical results, and to Mikkel Meyer Andersen for providing a cleaned version of the database of \citet{purps:2014}.
This research was supported by the Swiss National Science Foundation, through grants no.\ 105311-144557 and 10531A-156146, and carried out in the context of a joint research project, supervised by Franco Taroni (University of Lausanne, Ecole des sciences criminelles), and Richard Gill (Mathematical Institute, Leiden University).

%\bibliographystyle{agsm}
%\bibliography{/Users/esc/Dropbox/articoli/bibliografiaGiulia}
%
%

\end{document}